\documentclass[11pt]{article}

\usepackage[utf8]{inputenc}

\usepackage{url,amsfonts, amsmath, amssymb, amsthm,color, enumerate, multicol}

\usepackage{blkarray}

\usepackage{cite}             

\usepackage{multirow}

\usepackage{url}
\allowdisplaybreaks
\usepackage{hyperref}
\hypersetup{colorlinks=true,linkcolor=blue,citecolor=blue}

\renewcommand{\vec}[1]{\mathbf{#1}}
\usepackage[dvipsnames,usenames]{xcolor}
\newcommand{\explain}[1]{\tag{\textcolor{gray}{#1}}}

\usepackage{color}
\usepackage{float}
\usepackage{bbm}
\floatstyle{boxed}
\newfloat{algorithm}{t}{lop}
\usepackage{tikz}
\usepackage{varwidth}

\usepackage{booktabs} 
\DeclareMathOperator{\Ex}{\mathbb{E}}

\usepackage{epigraph}
\usepackage{nicefrac}
\usepackage{mathtools,extarrows}

\usepackage{fullpage}
\usepackage{relsize}
\usepackage{subfigure}
\usepackage{graphicx}

\usepackage{multicol}
\usepackage{comment}
\usepackage{parskip}

\DeclarePairedDelimiterXPP\Aver[1]{\mathbb{E}}{[}{]}{}{
	
	#1
}

\newtheorem{lemma}{Lemma}

\newtheorem{definition}{Definition}
\newtheorem{proposition}{Proposition}

\author{Simina Br\^anzei \and  Nithish Kumar \and Gireeja Ranade}
\title{Phase Transitions of Diversity in  Stochastic Block Model Dynamics\footnote{The authors are in  alphabetical  order. This work was done in part while S. Br\^anzei was visiting the Simons Institute for the Theory of Computing. S. Br\^anzei was  supported in part by US National Science Foundation CAREER grant CCF-2238372. N. Kumar was supported in part by US National Science Foundation grants   CCF-1910411 and  CCF-2228814. G. Ranade was  supported in part by  US National Science Foundation  CAREER grant ECCS-2240031, and also by  a Google Inclusion Award. }}
\author{Simina Br\^anzei\footnote{Purdue University. E-mail: simina.branzei@gmail.com.}
	\and 
	Nithish Kumar\footnote{Purdue University. E-mail: kumar410@purdue.edu.}
	\and 
	Gireeja Ranade\footnote{University of California Berkeley. E-mail: ranade@eecs.berkeley.edu.}
}
\date{}

\begin{document}
	
	\maketitle

	\begin{abstract}
		This paper proposes a stochastic block model with dynamics where the population grows using preferential attachment.  Nodes with higher weighted degree are more likely to recruit new nodes, and nodes always recruit nodes from their own community. This model can capture how communities  grow or shrink based on their collaborations with other nodes in the network, where an edge represents a collaboration on a project.
		
		Focusing on the case of two communities, we  derive a deterministic approximation to the dynamics and characterize the phase transitions for diversity, i.e. the parameter regimes in which either one of the communities  dies out  or the two communities reach parity over time.
		
		In particular, we find that the minority may vanish when the probability of cross-community edges is low, even when cross-community projects are more valuable than projects with collaborators from the same community.  
	\end{abstract}

	\section{Introduction}

	Academic institutions and large technology companies  have been increasing their focus on hiring diverse faculty members and employees, particularly at senior levels
	~\cite{bradley2018impact,salem2022don}. Multiple papers in the broader computer science community have addressed issues around fair hiring or selection processes (e.g.~\cite{salem2022don, kleinberg2016inherent, mitchell2021algorithmic, suhr2021does}). This paper considers what happens after hiring. We focus on the recruitment of candidates and growth of minority communities within organizations over time, in the context of the collaboration networks in the organization.

	An important part of recruitment and retention is the climate of an organization, which is heavily influenced by the interpersonal network of the members. Such networks and social capital~\cite{bourdieu2018forms} have a tremendous impact of the success of an individual~\cite{putnam2000bowling}. ~\cite{li2019early} shows how collaborations with top scientists can have significant impacts on the  careers of junior scientists. Additionally, network ties are known to have impact on employment~\cite{calvo2005social}, career-advancement~\cite{tharenou1997explanations, cullen2023old}, and  educational outcomes~\cite{calvo2009peer, kohli2023inclusive}. 
	
	The structures of both personal and professional networks are strongly influenced by the homophily principle, which states that a person is more likely to form connections with individuals that have shared characteristics~\cite{mcpherson2001birds, smith2014social}. Homophily can contribute to occupation-based segregation (i.e., individuals of the same gender/ethnicity seem to cluster into similar jobs) as well as establishment-based segregation (i.e. individuals of similar ethnicity/gender work at the same institutions), which in turn can contribute to wage gaps~\cite{petersen1995separate}. Thus, people are more likely to recruit members of their own communities to their organizations. Meanwhile, individuals are more likely to join organizations where those similar to themselves have been successful; for example,  newly recruited women often ask other women about their experiences in an organization. The work in \cite{guimera2005team} has shown that recruitment processes can affect network structure. In this paper, we consider how the network structure can affect recruitment.

	We consider a stochastic block model~\cite{HLL83} with two communities, red and blue. Edges in the graph represent collaborations on projects between agents, and the weight associated with the edge represents the value or impact of the project. Multiple new nodes join the network at every time step. To capture homophily, nodes recruit other nodes of the same color (red/blue). The population grows using preferential attachment~\cite{barabasi1999emergence, barabasi2002evolution}, where the attachment is based on the weighted degrees of the nodes.	
	 Thus, the more successful a community is---as measured by the  weight of its projects---the more new nodes from that community join the graph.
	\subsection{Related work}
	
	\paragraph{Network formation.} We are motivated to model scientific collaboration networks, the growth of which has been studied in~\cite{barabasi2002evolution, huang2008collaboration}. Others have examined how identity (e.g. gender) plays a role in this collaboration network~\cite{agarwal2016women,gertsberg2022unintended}. 
	
	Different strategic aspects of collaboration networks have been considered; in particular,~\cite{jackson2003strategic} considers a model where agents prefer to collaborate with nodes that have few edges (i.e. other collaborations), since this may indicate they have more available time to work on a joint project. 
	
	In contrast,~\cite{calvo2009peer} considers a social network model where edges in the network have weights, and the payoff to an agent is a function not only of their own effort but also of the effort of their peers (i.e. those they are connected to in the network). Our model also considers weighted edges, and the payoff to an agent is the sum of the weights across all its edges at a particular time, but we do not consider strategic aspects. \cite{kohli2023inclusive} addresses the problem of circumventing in-group bias that may appear in collaborations among university students on course projects.

	The classical Barabasi-Albert algorithm for  
	generating random scale-free networks using preferential attachment  was introduced in \cite{barabasi1999emergence}. 
	\cite{SHC20} considered a model of biased networks related to ours, where each node belongs to a community (e.g. red or blue). Nodes form connections according to the following principles:
	\begin{description}
		\item[$\; \;$] Community affiliation: When  a node enters the network, it
		chooses label red  with a probability $r$, and
		blue  with probability $1 - r$; 
		\item[$\; \;$] Rich-get-richer: Each arriving node  chooses to connect to some
		other node according to preferential attachment (i.e. with probability proportional to that node’s degree);  
		\item[$\; \;$] Homophily: If the two nodes from the previous step have
		the same label, an edge is formed; otherwise, the new node
		accepts the connection with probability $p$, where $0 < p < 1$,
		and the process is repeated until an edge is formed.
	\end{description}
	
	In this setting, \cite{SHC20} considered the problem of seeding: how to introduce early adopters (e.g. nodes from under-represented communities) in the network so that fairer outcomes are achieved in the long term. Moreover, \cite{SHC20} prove the existence
	of an analytical condition in which diversity acts as a catalyst
	for efficiency and analyze experimentally scientific networks using DBLP data. 
	
	The seeding or influence maximization problem was  posed in \cite{DR01}. An algorithmic formulation and analysis was given in \cite{KKT03}, which  identified a class of naturally occurring  processes where the influence maximization problem reduces to
	maximizing a submodular function under a cardinality constraint. Due to the desirable properties of submodular functions, this connection leads to efficient approximation algorithms for influence maximization. Learning to influence from observations of cascades was studied in \cite{BIS17}.
	For an extensive survey on  network formation and its connections to games and markets, see  the book \cite{KE12}.

	\paragraph{Stochastic block models.} The stochastic block model is a generative model for random graphs introduced in \cite{HLL83}, where there is a set $[n] = \{1, \ldots, n\}$ of nodes, each belonging to a community $C_1, \ldots, C_k$, and a $ k \times k$ symmetric probability matrix $\vec{p}$. A random graph is generated as follows: for each pair of vertices $u \in C_i$ and $v \in C_j$, generate the edge $(u,v)$ with probability $p_{i,j}$.  A classic question is community recovery: given one or more graphs generated in this fashion, recover the underlying communities of the nodes.
	Weighted stochastic block models were formulated in~\cite{aicher2013adapting,AJC14,Peixoto18}.
	A survey on  stochastic block models can be found in \cite{abbe2017}.

	Our model is informed by works such as~\cite{xu2014dynamic,ludkin2018dynamic} that consider time-evolving stochastic block models, with a fixed set of nodes. Our model explicitly considers the addition of new nodes at each time step, and these new nodes join using preferential attachment, similar to~\cite{collevecchio2013preferential}.

	\paragraph{Biases in algorithmic decision making.} 
	Studies have investigated how automated decision-making can have positive feedback effects of reinforcing biases over time~\cite{lowry1988blot, barocas2016big, liu2020disparate, ensign2018runaway}. Our model shows that  without high-value cross-community collaborations, there can be a  feedback effect where the majority community  grows while the minority community decays.  \cite{liu2018delayed} observed that common static fairness criteria can worsen long-term well-being.  To show this, \cite{liu2018delayed} considered a one-step feedback model to analyze how (classification) decisions change the underlying population over time.

	\paragraph{Urn models.} Our model is also strongly related to the Polya urn model, which first appeared in \cite{EP23,Polya31} and features an urn containing red and blue balls. The urn evolves in discrete time steps. At each step, a ball is drawn uniformly at random, its color is observed, and then it is returned to the urn together with a new ball of the same color. Questions of interest include the  ratio of red to blue in the long term and the stochastic path leading to it.
	A generalization is where the urn has balls of  $k$ colors. Whenever a ball of color $i$ is drawn, $A_{i,j}$ balls of color $j$ are brought to the urn, for each $j \in \{1, \ldots, k\}$, where $\vec{A} = \{ A_{i,j} \}_{i,j \in [k]}$ is the ``schema''. A book on the Polya urns with schemas can be found in    \cite{mahmoud2008polya}.
	
	Another generalization of the Polya urn  was analyzed in~\cite{HLS80}, where there is  an urn with some  initial number of red and blue balls. In each iteration, if the fraction of red balls is $x$, then with probability $f(x)$ a red ball is added next and with probability $1-f(x)$ a blue ball is added \footnote{ That is, $x$ is  the number of red balls divided by the total number of balls.}, where    $f : [0,1] \to [0,1]$ is a function.  Another generalization, where the update function can be time-dependent,  was analyzed in \cite{Pemantle90}. Nonlinear Polya urn processes with fitness were studied in ~\cite{BFRT16}. Our model can be seen as an urn which is time dependent (since the probability of a red ball being added depends both on the fraction of red and the total number of balls, which is equivalent to time dependence), but where multiple balls are added at each point in time.

	The Polya urn model  illustrates the rich get richer phenomenon, also known as the ``Matthew effect'' or  ``cumulative  advantage''.   \cite{Pemantle2007} surveys urn models together with their  analysis techniques and applications, such as   growth of     social or neural networks, dynamics in games and  markets~\cite{Pemantle2007}. A probabilistic model for the early phase of neuron growth was proposed in~\cite{KK01}, where  growing neurites compete with each other and longer ones have more chances to grow. 

			\paragraph{Evolution of populations.} The evolution of populations has been studied  under  forces such as  genetic drift, selection, mutation, and migration~\cite{nowak2006evolutionary}.  Some of these processes can be seen as urn models, where the composition of different types of marbles (gene variants) in a jar changes over time.
	In the Wright-Fisher process, there is an initial generation with some pool of genes. Then new generations appear such that  each copy of a gene found in the new generation is drawn independently at random from the copies of the gene in the old generation; the generations do not overlap. In the Moran process~\cite{moran1958random}, the generations are overlapping. The self-organization of matter from the lens of evolving macro-molecules was studied in~\cite{eigen71}.  
	For an overview of evolutionary game dynamics, see~\cite{hofbauer2003evolutionary}. 
	The convergence, mixing time, and computational questions arising in the evolution of populations have also been studied (see, e.g., \cite{dixit2012finite,PSV16}).
	
	\paragraph{Inequality in games and markets.} At the level of countries, a global trend of escalating economic inequality has been reported ~\cite{UN_inequality}. Extensive discussion and modeling of this phenomenon can be found in  \cite{B94,PG15}.  
	To counter the mounting accumulation of wealth that seems intrinsically tied to capitalist economies, \cite{Piketty14} proposed strategies like substantial wealth and inheritance taxes.
	Mathematical analyses showing the emergence of inequality in games and markets  include  \cite{BMN18} and \cite{GKMPP19}.
	In  non-atomic congestion games, \cite{GKMPP19} showed that  introducing an optimal mechanism may settle or even exacerbate social inequality.
	
	In networked production economies where players  use proportional response dynamics to update their strategies,  \cite{BMN18}  showed that  players get differentiated over time into two classes: the ``rich'' (who participate in the most efficient production cycle) and the ``poor'' (who do not); moreover,  the inequality gaps    between the rich and the poor players grow unboundedly over time. 
	{In our model, biases of the different nodes (players) may cause a similar rich-get-richer phenomenon, even when diverse teams have better performance.}

	\subsection{Main contributions}
	This paper introduces a stochastic model for network growth in Sec.~\ref{sec:model} that captures how network ties can influence community growth in the network. The model has two main parameters: 
	the probability of forming  an edge (collaboration), captured by the probability matrix $\vec{p}$; and 
	the value/weight of an edge, captured by the weight matrix $\boldsymbol{\zeta}$.  Both the probability of forming an edge and its value depend  on the communities of its endpoints.

	We derive (Sec.~\ref{sec:detapprox},~\ref{sec:detsys}) and analyze (Sec.~\ref{sec:detanalysis}) a deterministic approximation of the stochastic model and show that it closely mirrors the behavior of the stochastic model. We find that the deterministic model with two communities has three fixed points, at $x = 0, 1/2, 1$, where $x$ is the fraction of red in the system. 
	
	We identify a key ratio, $\rho$, which depends on $\vec{p}$ and $\boldsymbol{\zeta}$. 
	When $\rho > 1$, only $0$ and $1$ are stable fixed points, which means that one of the communities will necessarily vanish. This is possible if the probability of cross-edges is  low, even if cross-community projects are more valuable than monochrome projects (i.e. projects with collaborators from the same community).  The situation is reversed when $\rho < 1$.

	\vspace{0.1in}
	\noindent
	{\bf Main Theorem (informal):} Consider the stochastic block model dynamics with two communities of Definition \ref{def:stochastic}. The dynamic starts with a given graph and proceeds to refresh its edges while  new nodes arrive through  preferential attachment {based on  weighted degree} in each round.
	
	The resulting dynamical system admits a deterministic approximation,  which unfolds   as a function of  the  minority's fraction. For the deterministic approximation system we have: 
	\begin{enumerate}[(i)]
		\item The fixed points of the system are at  $0$, $1/2$, and $1$.
		\item There exists a  parameter, $\rho$, dependent on the probability and weight matrix of the underlying stochastic block model, which determines the behavior of the system:
		\begin{itemize}
			\item If $\rho > 1$, the minority  vanishes in the limit, with the fixed points at $0$ and $1$ being stable and the fixed point at  $1/2$  unstable.
			\item If $\rho = 1$, the minority fraction remains constant throughout time.
			\item If $\rho < 1$, the minority eventually achieves parity, with  the fixed point at $1/2$ being stable and the fixed points at $0$  and $1$ unstable.
		\end{itemize}
	\end{enumerate}

	\section{ Model} \label{sec:model}
	
	We  define the base model, which generates a random graph through a  weighted stochastic block model. Then we  describe the growth dynamics, where new nodes arrive  while the edges are refreshed.

	\subsection{ Weighted Stochastic Block Model} 
	
	\smallskip 
	
	We consider the following random graph generation process. 
	Each node has a color, red or blue. 

	\begin{definition}[Stochastic Block Model]  \label{def:stochastic_block_model}
		Let $[n] = \{1, \ldots, n\}$ be a set of nodes, so that each node $i$ has color $c_i \in \{R,B\}$, where $R=1$ is red and $B=2$ is blue. Let  $\vec{p}=  \begin{bmatrix}
			{a} & {b} \\
			{b} & {a}
		\end{bmatrix}$  and $ 
		\boldsymbol{\zeta}=  \begin{bmatrix}
			\alpha & \beta  \\
			\beta & \alpha 
		\end{bmatrix}$  
		{where} $a,b,\alpha,\beta \in \mathbb{R}_+\,.$
		
   Generate a weighted random graph on $[n]$, so that  
for each pair of nodes $i,j \in [n]$ the weight of edge $(i,j)$ is:
		\begin{itemize} 
			\item With probability ${p_{c_i,c_j}}/{n}$, let      $w_{i,j} := \zeta_{c_i,c_j}$.  
			\item With  remaining probability, let  $w_{i,j} :=0$. 
		\end{itemize}
	\end{definition}
	
	We refer to $\vec{p}$ as the probability matrix and $\boldsymbol{\zeta}$ as the weight matrix. The matrix $\vec{p}$ dictates edge existence, while the weights of edges that exist are  dictated by  $\boldsymbol{\zeta}$.
	In  general formulations of the weighted stochastic block model, the weight of an edge is drawn from a distribution that depends on the communities of the endpoints. We focus on the special case where this distribution is supported at a single point.

	\noindent \textbf{\textit{Interpretation of the model.}}	We  illustrate the stochastic block model through a scenario of agents working on projects in an organization:
	
	\begin{quote}
		\emph{There is a set of agents $[n] = \{1, \ldots, n\}$ that belong to one of two communities, red or blue. Node  $ i $ has community  $c_i \in \{R,B\}$. }
		
		\emph{Edges between agents indicate collaborations on projects. Agents may collaborate in pairs, or work individually on projects. Each pair of agents $(i,j)$ can collaborate on at most one project.  Self loops are allowed: edge $(i,i)$ means agent $i$ works alone on a project. }
		
		\emph{The  probability that agents $i$ and $j$ collaborate on a joint project is $p_{c_i, c_j}/n$.
		Thus the probability of a project between agents $i$ and $j$ depends on the communities of  $i$ and $j$, and the total number of agents.}
		
		\emph{If agents $i$ and $j$   collaborate on a project, i.e. form edge $(i,j)$, then the weight  of the resulting project is $\zeta_{c_i, c_j}$. 
		The weight can indicate the value/impact of the project and depends only on the communities of agents $i$ and $j$.}
		
		\emph{An agent's success is determined by the collective weight of all of their projects, both collaborative and individual.}
	\end{quote}

	The dependence on the total number of nodes in the probability $p_{c_i,c_j}/n$ of generating edge $(i,j)$ ensures the expected total number of projects that an agent works on remains bounded even as the size of the graph grows. 

	\subsection{ Stochastic Block Model Dynamics}
	
	We first motivate our model for the stochastic block  dynamics,  then give the mathematical definition.

	\noindent \textbf{\em Motivation.}	The dynamics capture the connection between recruitment/hiring, and the existing climate, collaborations and related success in the organization through the following features: 
	\begin{quote}
		\emph{Hiring/recruitment happens at  discrete time intervals (e.g. every year for an academic institution).}
		\emph{Collaborative projects are refreshed on a similar time scale (e.g. every year). Thus, at every time step, edge connections between nodes are refreshed. }
		
		\emph{Red  agents recruit other red  agents, and blue agents recruit other blue agents. }
		
		\emph{The higher the weight of an agent the more likely they are to recruit a new agent from their community. This captures two things: more successful agents are more likely to recruit other agents (e.g. by giving more frequent talks where they advertise their organization); and  agents are more likely to join organizations where  agents from their community are happy and successful.} 
		
		\emph{New nodes arrive arrive by  preferential attachment. Red  nodes arrive at a rate that is proportional to the total weight of projects in the red community at that time; and similarly for blue agents.}
		
		\emph{Note  the number of red nodes recruited at a time does not depend only on the existing number of red nodes, but on the total weight of their  projects (i.e. their weighted degree). }
	\end{quote}

 Let $w_u = \sum_{v \in V} w_{u,v}$ denote the weighted degree of vertex $u \in V$ and $W = \sum_{u \in V} w_u$  the sum of weighted degrees of all nodes. Sampling a vertex with likelihood given by their weighted degree  means that vertex $ u \in V$ is selected with probability ${w_u}/{W}$.
	
	Next we define the stochastic block  dynamics, which  generates a sequence of graphs $\{G_t\}_{t=1}^{\infty}$. Given $G_{t-1}$, an intermediate random graph $G_{t-1}^+$ is obtained by refreshing the edges of $G_{t-1}$,  where an edge is refreshed by  drawing it again  according to the stochastic block model with probability matrix $\vec{p}$ and weight matrix $\zeta$.
	Then, we grow $G_{t-1}^+$ by having a number $(m_t)$ of nodes join using preferential attachment to obtain $G_t$.
	
	\medskip 
	\begin{definition}[\textbf{Stochastic Block Model Dynamics}]  \label{def:stochastic}
		Let ${V}_0 = \{1, \ldots, n_0\}$ be an initial set of vertices,   each  red or blue. 
		At every time  unit $t = 1,2,\ldots$, the next steps take place:
		\begin{enumerate}[1.]
			\item Generate an intermediate  graph  $G_{t-1}^+$ on vertices ${V}_{t-1}^+ := {V}_{t-1}$ from  the  stochastic block model with probability matrix $\vec{p}$ and weight matrix $\boldsymbol{\zeta}$ \textcolor{gray}{(\emph{e.g. the agents in ${V}_{t-1}$ start pairwise/solo projects})}. 

   \item Initialize ${G}_t = {G}_{t-1}^+$. Independently sample with replacement  $m_t$  vertices $v_1, \ldots, v_{m_t} \in {V}_{t-1}^+$,  each with likelihood given by their weighted degree.  Each sample $v_i$  brings in a new vertex $u_i$ with the same color and they form an edge.
			Vertex $u_i$ and  edge $(v_i, u_i)$ are added to ${G}_{t}$ \textcolor{gray}{(\emph{e.g., agent $v_i$ recruits $u_i$})}.
		\end{enumerate} 
		
		{\emph{Constant fraction of arriving nodes:}} We  focus on the scenario where the  number of nodes arriving is a constant fraction $\lambda \in (0,1)$ of the existing population. That is, the random variable $m_t$ for the number of arriving nodes is  $m_t \in \{\lfloor \lambda \cdot n_{t-1} \rfloor, \lceil \lambda \cdot n_{t-1} \rceil\}$ such that $\Ex[m_t] = \lambda \cdot n_{t-1}$, where  $n_{t-1}$ is the number of nodes in ${G}_{t-1}$.  \hfill  $\square$
	\end{definition}

	Constant fraction of arriving nodes  models the early phases of a growing organization. Generally speaking, other regimes for $m_t$ could also be of interest.
	
	W.l.o.g.,  red is the initial minority, so  $n_{0}^R \leq n_0^{B}$.

	\begin{figure*}[h!]
		\centering
		\subfigure[Set $V_0$.]
		{
			\includegraphics[scale=.82]{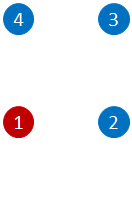}
			\label{fig:step1}
		}\;\;\;\;\;
		\subfigure[Graph $G_0^+$.]
		{
			\includegraphics[scale=.82]{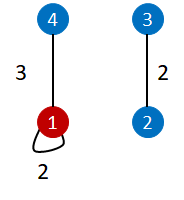}			\label{fig:step2}
		}\;\;\;\;\;
		\subfigure[Graph $G_0^+$; samples $1,3,3$.]
		{
			\includegraphics[scale=.82]{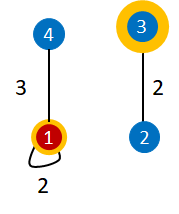}			\label{fig:step3}
		}\;\;\;\;\;
		\subfigure[Graph ${G}_1$. Each sample brings in a new node of the same color and they form an edge. That is, $1$ brings  $5$, the first occurrence of $3$ brings $6$, and the second occurrence of $3$ brings $7$.]
		{
			\includegraphics[scale=.82]{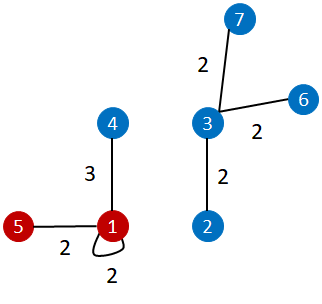}			\label{fig:step4}
		}
		\;\;
		\subfigure[Graph $G_1^+$. New edges are drawn between the vertices of $G_1$ to generate $G_1^+$.]
		{
			\includegraphics[scale=.82]{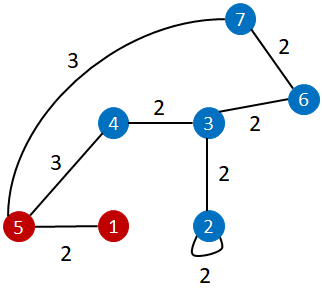}			\label{fig:step5}
		}
		\caption{Initial steps of the stochastic block model dynamics from  Definition~\ref{def:stochastic} illustrated. Suppose $\vec{p} = [[0.75, 0.25], [0.25, 0.75]]$,  the weights are $\boldsymbol{\zeta} = [[2, 3], [3, 2]]$, and $\lambda = 3/4$. Figure (a) shows the initial set of nodes $V_0 = \{1, 2, 3, 4\}$. Figure (b) shows a  realization of the random graph $G_0^+$ generated using the stochastic block model on $V_0$ with the parameters. Figure (c) shows the samples $(1, 3, 3)$ drawn from the vertices of $G_0^+$ according to the weighted degree distribution of $G_0^+$. Figure (d) shows how each sample $v_i$ brings in a new node $u_i$, which copies $v_i$'s color and they form an edge. The graph  obtained is $G_1$. Figure (e) shows graph $G_1^+$, generated using the  same parameters on the vertices of $G_1$. 
		}
		\label{fig:sequence_graph_generation}
	\end{figure*}

	\begin{figure*}[h!]
		\centering
		\subfigure[]
		{
			\includegraphics[scale=1.6]{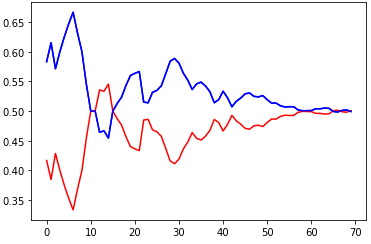}
			\label{fig:parity}
		}
		\subfigure[]
		{
			\includegraphics[scale =1.6]{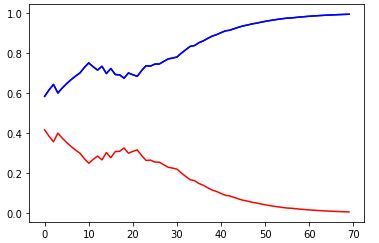}
			\label{fig:vanishes}
		}
		\caption{Trial runs of the stochastic block dynamics from Definition~\ref{def:stochastic}  with $11$ initial nodes. Red is the initial minority, starting with $5$ nodes in both cases. The $X$ axis shows time and the $Y$ axis shows the value of the fraction of red nodes. In (a) the fraction of red nodes reaches 50\% in the limit, while in (b) it goes to zero in the limit. The probability matrix is $\vec{p} = [[0.75, 0.25], [0.25, 0.75]]$ in both cases (a) and (b). The fraction of nodes arriving in each round is $\lambda = 0.1$. The weight matrix is $\boldsymbol{\zeta} = [[1, 100], [100, 1]]$ in (a) and $\boldsymbol{\zeta}  = [[100, 1], [1, 100]]$ in (b). 
		}
		\label{fig:stochastic_block_dynamics}
	\end{figure*}
	
	Figure~\ref{fig:sequence_graph_generation} illustrates the initial steps of the dynamic for a small number of initial nodes.
	A picture with trial runs of the dynamic that shows {the fraction of red nodes over time} is included in Figure~\ref{fig:stochastic_block_dynamics}.

	\newpage 
	
	\section{Deterministic Approximation}  \label{sec:detapprox}
	In this section we deduce a deterministic approximation for the random system of Definition~\ref{def:stochastic} and prove that it closely mirrors the stochastic model when the fraction of red (i.e. the minority fraction) is not arbitrarily close to zero.
	
	We first introduce some notation that will be useful. 		In the sequence of graphs from Definition~\ref{def:stochastic}, we index the vertices  $1, 2, \ldots$.   Thus if at some point the existing vertices are  $\{1, \ldots, k\}$, then  the next arriving node is denoted $k+1$.
	
	In the setting of Definition~\ref{def:stochastic}, for each time step $t \in \mathbb{N}$:
	\begin{itemize}
		\item  Let $\mathcal{F}_t$ be the history at the end of round $t$ for the process $\{G_t\}_{t=0}^{\infty}$.
		\item Let $\phi_t = \frac{n_{t}^R}{n_t}$ be the fraction of red nodes in ${G}_t$.
		\item Let $w_{i,j}^+(t-1)$ be the weight of edge $(i,j)$ in the graph $G_{t-1}^+$. 
		Let $R_t$ be the weight of color red and $B_t$ the weight of color blue, respectively, in graph $G_{t-1}^+$.  That is, 
		\begin{align} 
			R_t & = \sum_{i \in [n_{t-1}]: c_i = R} \; \; \sum_{j \in [n_{t-1}]} w_{c_i, c_j}^+(t-1)\,.  \\
			B_t & = \sum_{i \in [n_{t-1}]: c_i = B} \; \;  \sum_{j \in [n_{t-1}]} w_{c_i, c_j}^+(t-1)\,.
		\end{align}
		\item Let $m_t^R$ and $m_t^B = m_t - m_t^R$ be the number of red and blue nodes, respectively, arriving at time step $t$.
	\end{itemize}

	Next we prove there exists a function $\Gamma: [0,1] \to [0,1]$ such that   if $\phi_{t-1}$ (i.e. the fraction of the minority at time $t-1$) is above $1/\sqrt{n_{t-1}}$ and the initial number ($n_0$) of nodes is large enough, then 
	\begin{align}
		& \bigl(1-o(1)\bigr) \left(  \frac{ \phi_{t-1} + \lambda  \cdot  \Gamma(\phi_{t-1}) }{1 +  \lambda}\right) \leq \Ex[\phi_{t} \mid \mathcal{F}_{t-1}]  \leq  \bigl(1 + o(1) \bigr)  \left(  \frac{ \phi_{t-1} + \lambda  \cdot  \Gamma(\phi_{t-1}) }{1 +  \lambda} \right), \notag
	\end{align}
	where 
	\begin{align} \Gamma(x) & = \frac{{a \alpha}  \cdot  x^2   + {b \beta}  \cdot  x  \cdot (1 - x)}{ {a \alpha}  \cdot  x^2   + 2   {b \beta}  \cdot x  \cdot (1 - x) +  {a \alpha} \cdot  (1 - x)^2} \,. \label{eq:gamma}
	\end{align}
	This will allow us to conclude that the dynamical system evolves approximately as follows: 
	\begin{align}
		\Ex[\phi_{t} \mid \mathcal{F}_{t-1}] \approx \frac{ \phi_{t-1} + \lambda  \cdot  \Gamma(\phi_{t-1}) }{1 +  \lambda}\,. \label{eq:approximation_rough_estimate}
	\end{align}
	\begin{lemma} \label{lem:epsilon_fraction_bounded_away}
		Suppose  there is a constant $\epsilon \in (0, 1/2)$ so that  
		\begin{align} 
			\left(\frac{1}{n_{t-1}}\right)^{\frac{1}{2}-\epsilon} \leq \phi_{t-1}  \leq \frac{1}{2} \,. \label{eq:condition_red_bounded_away}
		\end{align}
		Then there exists a constant $C = C(\epsilon, a, \alpha, b, \beta)$ so that when $n_0 \geq C$, 
		\[
		\mathcal{L}(n_{t-1}) \leq  \Ex[\phi_t \mid \mathcal{F}_{t-1}] \leq \mathcal{U}(n_{t-1}),
		\]
		where $\mathcal{L}, \mathcal{U}: \mathbb{N} \to \mathbb{R}$   are defined by  
		\begin{align}
			\mathcal{U}(n) & = \frac{ \phi_{t-1} + \frac{\lceil \lambda \cdot n\rceil}{n} \cdot \left(1 + {1}/{n^{\epsilon/5}}\right) \cdot  \Gamma(\phi_{t-1}) }{1 + \frac{\lfloor \lambda n \rfloor}{n}}  \mbox{ and } \; 
			\mathcal{L}(n)  = \frac{ \phi_{t-1} + \frac{\lfloor \lambda \cdot n\rfloor}{n} \cdot \left(1 - {1}/{n^{\epsilon/5}}\right) \cdot  \Gamma(\phi_{t-1}) }{1 + \frac{\lceil \lambda n \rceil}{n }} \,. \label{def:mathcal_U_and_L} 
		\end{align}
	\end{lemma}
	\begin{proof} 
		For a fixed round $t$, each arriving node has the same probability of being red. The probability that one arriving node  is red  equals $R_t/(R_t + B_t)$. Taking expectation, since the number of arriving nodes is $m_t \in \{ \lfloor \lambda  \cdot  n_{t-1} \rfloor, \lceil \lambda \cdot  n_{t-1} \rceil\}$, we  have 
		\begin{align} 
			\sum_{i=1}^{\lfloor \lambda \cdot n_{t-1}\rfloor } \Ex\left[ \frac{R_t}{R_t + B_t} \mid \mathcal{F}_{t-1} \right]  & \leq \Ex\left[m_t^R \mid \mathcal{F}_{t-1}\right]  \leq \sum_{i=1}^{\lceil \lambda \cdot n_{t-1}\rceil} \Ex\left[ \frac{R_t}{R_t + B_t} \mid \mathcal{F}_{t-1} \right] \,. \label{eq:bounds_sum_m_R_t}
		\end{align}
		
		By Lemma~\ref{lem:expected_ratio_approx_ratio_expectations}, we have 
		\begin{align}
			\left(1 - {1}/{n_{t-1}^{\epsilon/5}}\right) \cdot \frac{\Ex[R_t \mid \mathcal{F}_{t-1} ]}{\Ex[{R_t} + {B_t} \mid \mathcal{F}_{t-1}]}  \leq    \Ex\left[ \frac{R_t}{R_t + B_t} \mid \mathcal{F}_{t-1} \right] \leq  \left(1 + {1}/{n_{t-1}^{\epsilon/5}}\right) \cdot \frac{\Ex[R_t \mid \mathcal{F}_{t-1} ]}{\Ex[{R_t} + {B_t} \mid \mathcal{F}_{t-1}]}  \,. \label{eq:bounds_prob_red}
		\end{align}
		
		Combining \eqref{eq:bounds_sum_m_R_t} and \eqref{eq:bounds_prob_red}, we obtain that  under \eqref{eq:condition_red_bounded_away}, 
		\begin{align}
			\lfloor \lambda \cdot n_{t-1} \rfloor  \left(1 - {1}/{n_{t-1}^{\epsilon/5}}\right) \frac{\Ex[R_t \mid \mathcal{F}_{t-1} ]}{\Ex[{R_t} + {B_t} \mid \mathcal{F}_{t-1}]} & \leq  \Ex\left[m_t^R \mid \mathcal{F}_{t-1}\right]  \notag \\
			& \leq  \lceil \lambda \cdot n_{t-1}\rceil  \left(1 + {1}/{n_{t-1}^{\epsilon/5}}\right) \frac{\Ex[R_t \mid \mathcal{F}_{t-1} ]}{\Ex[{R_t} + {B_t} \mid \mathcal{F}_{t-1}]} \,.  \label{eq:intermediate_bounds_expected_m_R_t}
		\end{align}
		
		Since the graphs $G_{t-1}^+$ and $G_{t-1}$ have the same vertices, we have 
		\begin{align} 
			& \Ex[{R_t} \mid \mathcal{F}_{t-1}]   =  \left(n_{t-1}^{R}\right)^2 \cdot \frac{a \alpha}{n_{t-1}} + n_{t-1}^R n_{t-1}^B \cdot \frac{b \beta}{n_{t-1}}\; ; \notag  \\
			& \Ex[{B_t}  \mid \mathcal{F}_{t-1}]  = \left( n_{t-1}^B \right)^2  \cdot \frac{a \alpha}{n_{t-1}} + n_{t-1}^R n_{t-1}^B \cdot \frac{b \beta}{n_{t-1}}\,.  \label{eq:expected_red_t_and_blue_t}
		\end{align}

		Combining the identities in \eqref{eq:expected_red_t_and_blue_t} gives 
		\begin{align}
			\frac{\Ex[R_t \mid \mathcal{F}_{t-1}]}{\Ex[R_t+B_t \mid \mathcal{F}_{t-1}]}  
			& = \frac{{a \alpha}   \left(n_{t-1}^{R}\right)^2   + {b \beta}  \left(  n_{t-1}^{R} \cdot n_{t-1}^{B} \right)}{ {a \alpha}   \left(n_{t-1}^{R} \right)^2   + 2   {b \beta}  \left(  n_{t-1}^{R} \cdot n_{t-1}^{B} \right) +  {a \alpha}  \left( n_{t-1}^{B} \right)^2 } \notag \\
			& = \frac{{a \alpha}    \left(\phi_{t-1}\right)^2   + {b \beta}    \phi_{t-1}   (1 - \phi_{t-1})}{ {a \alpha}    \left(\phi_{t-1}\right)^2   + 2   {b \beta}    \phi_{t-1}   (1 - \phi_{t-1}) +  {a \alpha}   (1 - \phi_{t-1})^2}  \notag \\
			& = \Gamma(\phi_{t-1})\,. \label{eq:ratio_of_expectations_gamma_xi_tminus1}
		\end{align}
		The number of red nodes in $\mathcal{G}_t$ is  $n_t^R = n_{t-1}^R + m_t^R$, so $\phi_t = \frac{n_{t-1}^R + m_{t}^R}{n_{t}}$. By definition of the dynamic,  
		\begin{align}
			& n_{t} \in \Bigl\{ n_{t-1} + \lfloor \lambda n_{t-1} \rfloor,  n_{t-1} + \lceil  \lambda n_{t-1} \rceil \Bigr\} \mbox{ with } \Ex[n_t \mid \mathcal{F}_{t-1}] = (1 + \lambda)n_{t-1}\,. \label{eq:simple_facts_n_t_v2}
		\end{align}
		Using \eqref{eq:intermediate_bounds_expected_m_R_t},  \eqref{eq:ratio_of_expectations_gamma_xi_tminus1}, and \eqref{eq:simple_facts_n_t_v2}, we can upper bound the expected fraction of red at time $t$  when   \eqref{eq:condition_red_bounded_away} holds:
		\begin{align}
			\Ex[\phi_t \mid \mathcal{F}_{t-1}] & \leq   \frac{n_{t-1}^R + \Ex[m_{t}^R \mid \mathcal{F}_{t-1}]}{n_{t-1} + \lfloor \lambda n_{t-1} \rfloor } 
			= \frac{ \phi_{t-1} + \frac{\lceil \lambda \cdot n_{t-1}\rceil}{n_{t-1}} \cdot \left(1 + {1}/{n_{t-1}^{\epsilon/5}}\right) \cdot  \Gamma(\phi_{t-1}) }{1 + \frac{\lfloor \lambda n_{t-1} \rfloor}{n_{t-1}}}  \,. \label{eq:ub_on_update_rule}
		\end{align}
		
		Similarly, we can lower bound the expected fraction of red at time $t$ as follows when   \eqref{eq:condition_red_bounded_away} holds:
		\begin{align}
			\Ex[\phi_t \mid \mathcal{F}_{t-1}] & \geq    \frac{n_{t-1}^R + \Ex[m_{t}^R \mid \mathcal{F}_{t-1}]}{n_{t-1} + \lceil  \lambda n_{t-1} \rceil } 
			= \frac{ \phi_{t-1} + \frac{\lfloor \lambda \cdot n_{t-1}\rfloor}{n_{t-1}} \cdot \left(1 - {1}/{n_{t-1}^{\epsilon/5}}\right) \cdot  \Gamma(\phi_{t-1}) }{1 + \frac{\lceil \lambda n_{t-1} \rceil}{n_{t-1}}}   \,. \label{eq:lb_on_update_rule}
		\end{align}

		Combining \eqref{eq:ub_on_update_rule}, \eqref{eq:lb_on_update_rule}, \eqref{def:mathcal_U_and_L}, we obtain
		$\mathcal{L}(n_{t-1}) \leq  \Ex[\phi_t \mid \mathcal{F}_{t-1}] \leq \mathcal{U}(n_{t-1})\,.$
		This completes the proof of the lemma.
	\end{proof}

	The proof of the next lemma is deferred to the appendix.
	\begin{lemma}  \label{lem:expected_ratio_approx_ratio_expectations}
		In the setting of Lemma~\ref{lem:epsilon_fraction_bounded_away}, 
		\begin{align}
			\left(1 - \frac{1}{n_{t-1}^{\epsilon/4}}\right) \frac{\Ex[{R_t} ]}{\Ex[{R_t} + {B_t} ]}  < \frac{{R_t}}{{R_t} + {B_t}} 
			< \left(1 + \frac{1}{n_{t-1}^{\epsilon/4}}\right) \frac{\Ex[{R_t}]}{\Ex[{R_t} + {B_t}]} \,. \notag 
		\end{align}
		with  probability at least $1 - \frac{8}{e^{C_1 \cdot n_{t-1}^{\epsilon}}}$, where $C_1= \min\left\{ {a}/{12}, {b}/{6}\right\}\,.$ 
		Moreover, 
		\begin{align}
			\left(1 - \frac{1}{n_{t-1}^{\epsilon/5}}\right) \frac{\Ex[{R_t}]}{\Ex[{R_t} + {B_t}]}  < \Ex\left[\frac{{R_t}}{{R_t} + {B_t}} \right] 
			<  \left(1 + \frac{1}{n_{t-1}^{\epsilon/5}}\right) \frac{\Ex[{R_t}]}{\Ex[{R_t} + {B_t}]} \,. \notag 
		\end{align}
	\end{lemma}
	
	By Lemma~\ref{lem:epsilon_fraction_bounded_away}, when the fraction of the minority, $\phi_{t-1}$, is strictly above $1/\sqrt{n_{t-1}}$, we have  
	\[ 
	\mathcal{L}(n_{t-1}) \leq  \Ex[\phi_t \mid \mathcal{F}_{t-1}] \leq \mathcal{U}(n_{t-1}).
	\] 
	Moreover,   
	\begin{align} 
		\lim_{n \to \infty} \mathcal{L}(n) = \lim_{n \to \infty} \mathcal{U}(n) =  \frac{ \phi_{t-1} + \lambda  \cdot  \Gamma(\phi_{t-1}) }{1 +  \lambda} \,. \label{eq:limit_LU}
	\end{align}
	
	\section{Deterministic System} \label{sec:detsys}
	
	The approximation in \eqref{eq:approximation_rough_estimate}, which in the limit is given by the right hand side of \eqref{eq:limit_LU},   motivates us to introduce the following deterministic system to approximate the behavior of the stochastic system for large number of nodes.

	\begin{definition}[Deterministic System]  \label{def:deterministic}
		Let $\Phi_0 \in (0,1/2]$. Define   the sequence $\{\Phi_{t}\}_{t=0}^{\infty}$      such that  
		\begin{align} 
			\Phi_{t} =  \frac{\Phi_{t-1} + \lambda \cdot \Gamma(\Phi_{t-1})}{1 + \lambda },  
		\end{align} 
		recalling that $\Gamma(x)$ \begin{align}
			\Gamma(x) = \frac{{a \alpha}  \cdot  x^2   + {b \beta}  \cdot  x  \cdot (1 - x)}{ {a \alpha}  \cdot  x^2   + 2   {b \beta}  \cdot x  \cdot (1 - x) +  {a \alpha} \cdot  (1 - x)^2} \,. \notag 
		\end{align} 
	\end{definition}
	Figures~\ref{fig:deterministic__random_side_by_side_parity} and~\ref{fig:deterministic__random_side_by_side_vanish} show a  trajectory of the deterministic system of Definition~\ref{def:deterministic} together with the trajectory of  the stochastic system of Definition~\ref{def:stochastic}  started from the same initial configuration. 
	\begin{figure*}[h!]
		\centering
		\subfigure[Random system]
		{
			\includegraphics[scale = 1.47]{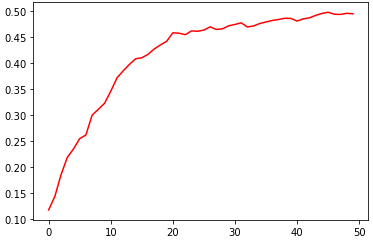}
			\label{fig:parity_rand}
		}
		\subfigure[Deterministic system]
		{
			\includegraphics[scale = 0.55]{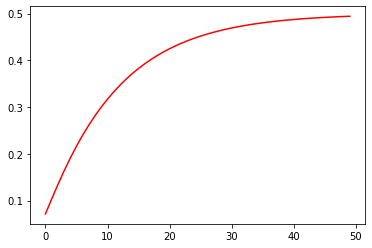}
			\label{fig:parity_det}
		}
		\caption{Figure (a) shows the random system with  $n=70$ initial nodes such that $n_R = 5$ are red. The probability matrix is $\vec{p}=  \begin{bmatrix}
				{0.75} & {0.25} \\
				{0.25} & {0.75}
			\end{bmatrix}$ and the weight matrix is $ 
			\vec{\omega}=  \begin{bmatrix}
				1 & 100  \\
				100 & 1 
			\end{bmatrix}$. We have $\rho \approx 0.03$.  The fraction of nodes arriving in each round is $\lambda = 0.1$.  Figure (b) shows the corresponding deterministic system, for the same initial parameters, as $n \to \infty$.}
		\label{fig:deterministic__random_side_by_side_parity}
	\end{figure*}

	\begin{figure*}[h!]
		\centering
		\subfigure[Random system]
		{
			\includegraphics[scale = 0.55]{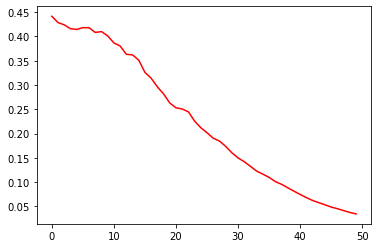}
			\label{fig:parity_rand_vanish}
		}
		\subfigure[Deterministic system]
		{
			\includegraphics[scale = 0.55]{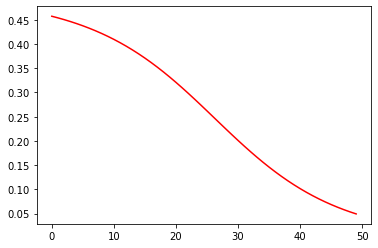}
			\label{fig:parity_det_vanish}
		}
		\caption{Figure (a) shows the random system with  $n=70$ initial nodes such that $n_R = 32$ are red. The probability matrix is $\vec{p}=  \begin{bmatrix}
				{0.95} & {0.05} \\
				{0.05} & {0.95}
			\end{bmatrix}$ and the weight matrix is $ 
			\vec{\omega}=  \begin{bmatrix}
				1 & 1.2  \\
				1.2 & 1 
			\end{bmatrix}$. We have $\rho \approx 15.8$.  The fraction of nodes arriving in each round is $\lambda = 0.1$.  Figure (b) shows the corresponding deterministic system, for the same initial parameters, as $n \to \infty$.}
		\label{fig:deterministic__random_side_by_side_vanish}
	\end{figure*}

	The deterministic model can be thought of as a version of the stochastic system where each node is an infinitesimal particle  (e.g. of some fluid such as water), which is partially red and partially blue instead of a discrete node with one color. 
	
	Recall that at each time step in the stochastic system, every arriving node $u$ becomes red with some  probability $p$ and blue with probability $1-p$. The deterministic system captures a similar dynamic, but each arriving node $u$ is an infinitesimal particle which has a mix of two colors:  red in proportion of $p$ and  blue in proportion of $1-p$. At each point in time, if there are currently $x$ units of fluid, then $\lambda x$ units are added in the next iteration. Explicit node edges are no longer considered in the deterministic approximation, but instead we have the product of the amount of red fluid and the amount of blue fluid to correspond to the product of red-blue edges.
	
	The variable of interest will be the evolution of the ratio red/blue over time.

	\section{Analysis of the Deterministic System} \label{sec:detanalysis}
	Next, we understand how the ratio of red to blue evolves over time in the deterministic system, by identifying the fixed points of the model, and their stability properties. In Lemma~\ref{lem:x_t_update_function} first compute a simplified update function for the system of Definition~\ref{def:deterministic}, which shows that the key parameters are the fraction of red at each point in time and the constant ratio $\rho := {a\alpha}/{(b \beta)}$. 
	
Proposition~\ref{prop1} shows that if $\rho > 1$ then the fraction of the minority (red) can die out regardless of the initial state.

	\begin{lemma} \label{lem:x_t_update_function}
		The  dynamical system in  Definition~\ref{def:deterministic} has update function, i.e. $f: [0,1] \to [0,1]$ {such that $\Phi_{t+1} = f(\Phi_{t})$} given by:
		\begin{align*} 
			f(x) = \frac{2x^3 (\rho - 1) -x^2 (\rho-1)(2- \lambda) + x(\rho + \lambda)}{(1+\lambda) \Bigl(2 x^2 (\rho - 1) - 2x (\rho - 1) + \rho \Bigr) }, \; \; \mbox{ where } \rho := a \alpha / (b \beta)\,.
			\label{def:update_function_deterministic} 
		\end{align*}
	\end{lemma}	
	\begin{proof}
		By Definition~\ref{def:deterministic}, we have 
		$\Phi_{t+1} =  \frac{\Phi_t + c \cdot \Gamma(\Phi_t)}{1 + \lambda }, $
		where $\Gamma(x)$ is as in~\eqref{eq:gamma}.		
		
		Consider the function $f:[0,1] \to [0,1]$ given by 
		$f(x) = \bigl({x + \lambda \cdot \Gamma(x)}\bigr)/\bigl({1+ \lambda}\bigr)$.
		Substituting $\Gamma(x)$ in the expression for $f(x)$ and re-arranging the terms gives 
		\begin{align*}
			f(x) & = \frac{x + \lambda \cdot \left( \frac{a \alpha \cdot x^2 + b\beta \cdot x \bigl(1 - x\bigr)}{a \alpha \cdot x^2  + 2 b \beta \cdot x \bigl(1 - x\bigr) +  a \alpha \cdot\bigl(1 - x\bigr)^2}\right) }{1 +  \lambda  }  \notag \\
			& = \frac{2x^3 (\rho - 1) -x^2 (\rho-1)(2- \lambda) + x(\rho + \lambda)}{(1+ \lambda) \left(2 x^2 (\rho - 1) - 2x (\rho - 1) + \rho \right) } 
		\end{align*} 
		where the equality is obtained by dividing by $b \beta$,  substituting $\rho = a \alpha / (b \beta)$ and simplifying.
	\end{proof}

	\begin{figure*}[h!]
		\centering
		\subfigure[The yellow line shows the function $f$ with $\rho = 12$ and $\lambda=2$. The blue line shows the identity function. The fixed points $0$ and $1$ of $f$ are stable while $1/2$ is unstable.]
		{
			\includegraphics[scale = 0.6]{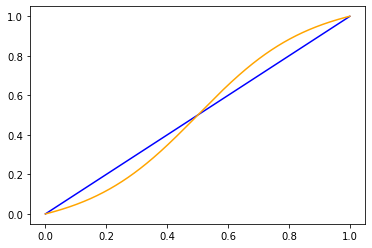}
		}
		\subfigure[The derivative $f'$ when $\rho = 12$ and $ \lambda =2$.]
		{
			\includegraphics[scale=0.6]{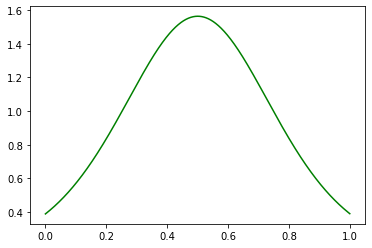}
		}
		\subfigure[The function $f$ when $\rho = 0.2$ and $\lambda=2$. The blue line shows the identity function. The fixed points $0$ and $1$ of $f$ are unstable while $1/2$ is stable.]
		{
			\includegraphics[scale = 0.6]{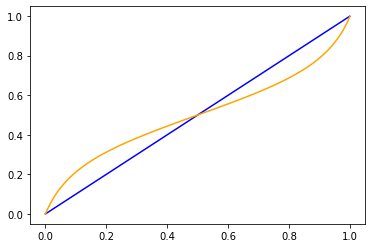}
		}
		\subfigure[The derivative $f'$ when $\rho = 0.2$ and $\lambda=2$.]
		{
			\includegraphics[scale=0.6]{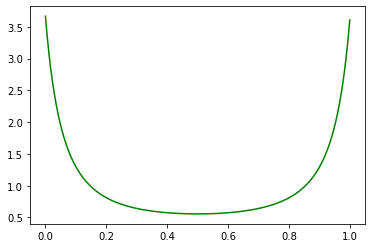}
		}
		\caption{The function $f$ and its derivative $f'$ for two values of $\rho$. In both cases  $\lambda = 2$. Note here we are choosing a $\lambda > 1$ to better illustrate the curvature of the function, but the shape remains similar for all $\lambda > 0$.}

		\label{fig:function_f_and_derivative}
	\end{figure*}

	Next, we identify the fixed points of the system. 
	
	\begin{lemma}
		\label{lem:simple_fixed_points_det}
		Consider the update function $f$ for the deterministic system.
		If $\rho \neq 1$ then the set of fixed points of  $f$ is   $\left\{0,\frac{1}{2},1\right\}$. If $\rho = 1$ then  $f(x) = x$ for all $x \in [0,1]$.
	\end{lemma}
	\begin{proof}
		To find the fixed points, we require that $f(x) = x$. Equivalently, 
		\begin{align*}
			& f(x) = x \iff \notag \\
			& 2x^3 (\rho - 1) -x^2 (\rho-1)(2- \lambda) + x(\rho + \lambda) =  (1+ \lambda) \left(2 x^3 (\rho - 1) - 2x^2 (\rho - 1) + \rho x \right) \iff \notag \\
			& \lambda (\rho - 1) x \bigl(2x -1 \bigr) \bigl( x-1\bigr) = 0 \,. 
		\end{align*}
		By definition, $\lambda > 0$.
		If $\rho = 1$, then $f(x) = x$ for any $x \in [0,1]$. If $\rho \neq 1$, then the roots are $0$, $1$, and $1/2$ as required.
	\end{proof}
	
	Finally, we consider the system convergence. This is visualized in  	Figure~\ref{fig:function_f_and_derivative}, which shows the function $f$ together with its  derivative for different values of $\rho$. 
	
	\begin{proposition}[Convergence of the deterministic system] \label{prop1}
		The deterministic system of Definition~\ref{def:deterministic} 
		\begin{itemize}
			\item  monotonically decreases to $0$ when $\rho > 1$.
			\item monotonically increases to  $1/2$ when $\rho < 1$.
			\item is constant when $\rho = 1$.
		\end{itemize}
	\end{proposition}
	\begin{proof}
		By Lemma~\ref{lem:simple_fixed_points_det}, if $\rho=1$ then the deterministic system is constant. 
		
		Suppose that $\rho \neq 1$. Then  for every $x \in (0,1/2)$, we  have 
		\begin{align}
			& f(x) < x \iff \notag  \\ 
			& \frac{2x^3 (\rho - 1) -x^2 (\rho-1)(2- \lambda) + x(\rho + \lambda)}{(1+ \lambda) \left(2 x^2 (\rho - 1) - 2x (\rho - 1) + \rho \right) } < x  \iff  \notag \\
			& 2x^2 (\rho - 1) - x (\rho -1 )(2 - \lambda) + (\rho + \lambda)  < (1+ \lambda)\left(2 x^2 (\rho - 1) - 2x (\rho - 1) + \rho \right) \iff \notag \\
			& x^2 \cdot 2 \lambda (\rho - 1)  - x \cdot 3 \lambda(\rho - 1) +  (\rho -1 ) > 0 \iff  \notag \\
			& (\rho - 1) \cdot (2x-1) (x-1)  > 0  \,.
		\end{align}
		Since $x \in (0,1/2)$, we get $(2x-1)(x-1) > 0$. Thus $f(x) < x$ for $\rho > 1$ and $f(x) > x$ for $\rho < 1$.
		
		Thus when $\rho > 1$, the system  monotonically decreases towards the fixed point $0$. When $\rho < 1$, the sequence of iterates monotonically increases. 
		
		We show that the fixed point reached must be  $1/2$ since the minority can never become the majority, i.e. $f(x) \leq  1/2$  for all $x \in [0,1/2]$. The inequality clearly holds for $x=0$ since $f(0) = 0$. When $x \in (0, 1/2]$, we have:
		\begin{align*} 
			& f(x) \leq  1/2 \iff  \frac{x (2x^2 (\rho - 1) -x (\rho-1)(2- \lambda) + (\rho + \lambda))}{(1+\lambda) \left(2 x^2 (\rho - 1) - 2x (\rho - 1) + \rho \right) }  \leq \frac{1}{2}   \,. 
		\end{align*}
		
		Since $x \in (0,\frac{1}{2})$, it suffices to show that 
		\begin{align} \label{eq:sufficient_x_strict_interior_v2}
			\frac{2x^2 (\rho - 1) -x (\rho-1)(2- \lambda) + (\rho + \lambda)}{(1+\lambda) \left(2 x^2 (\rho - 1) - 2x (\rho - 1) + \rho \right) }  \leq 1 \,. 
		\end{align}
		
		Cross-multiplying and simplifying in \eqref{eq:sufficient_x_strict_interior_v2}, 
		since $\lambda > 0$, it suffices to show that
		$$ x (\rho - 1) + 1 \leq 2x^2 (\rho - 1) - 2x (\rho - 1) + \rho.$$
		This inequality is true iff
		$ (\rho - 1) (2x -1 ) (x-1) \leq 0,$
		which is always true when $\rho < 1$ and $x \in (0,\frac{1}{2})$. This  completes the proof. \end{proof}

	Next we deduce the stability of the fixed points by analyzing the derivative of the function $f$.
	
	\begin{lemma} \label{lem:definition_derivative_f}
		The derivative of the update function $f$ is 
		
		\begin{equation*}
			f'(x) =  \frac{ 4 x^4 (\rho - 1)^2 -8 x^3 (\rho - 1)^2 +\rho (\lambda + \rho)   - 2 x^2 (\rho - 1) \Bigl(2 + \rho ( \lambda - 4) \Bigr) + 2 x \rho  ( \rho - 1)( \lambda - 2)}{(1 + \lambda) \Bigl(2 x^2 ( \rho - 1 ) -2 x ( \rho - 1) +  \rho \Bigr)^2}\,.
		\end{equation*}		
		In particular,  
		\begin{align} 
			f'(0) = f'(1) & = \frac{\lambda +\rho}{\lambda \rho + \rho} \; \; \mbox{ and } \; \; 
			f'({1}/{2}) = \frac{2 \lambda \rho + \rho + 1}{(1 + \lambda)(\rho+1)}\,. \label{eq:evaluation_f_prime_fixed_points}
		\end{align}
	\end{lemma}
	\begin{proof}
		The calculation to obtain the derivative  is standard and can be checked, e.g., with Mathematica~\cite{mathematica}.
		Evaluating it at the fixed points, we obtain \eqref{eq:evaluation_f_prime_fixed_points}.
		This completes the proof of the lemma.
	\end{proof}

	For a one dimensional deterministic system with update function $f$, a fixed point $x$ is stable if $|f'(x)| < 1$ and unstable if $|f'(x)| > 1$. 
	
	\begin{proposition}[Stability of fixed points]
		The fixed points  $\{0,1\}$ are stable  when $\rho > 1$ and unstable when $\rho < 1$. The fixed point  $1/2$ is stable when $\rho < 1$ and unstable when $\rho > 1$.
	\end{proposition}
	\begin{proof}
		By Lemma~\ref{lem:definition_derivative_f}, we have 
		\begin{align}
			& |f'(0)| = |f'(1)| = \left| \frac{ \lambda+\rho}{\lambda \rho + \rho}\right|  < 1 \notag   \\
			& \iff \frac{ \lambda +\rho}{\lambda \rho + \rho} < 1 \explain{Since $\lambda,\rho > 0$} \\
			& \iff \rho > 1\,. \notag 
		\end{align}
		Thus $0$ and $1$ are stable when $\rho > 1$. 
		Similarly,  $0$ and $1$ are unstable when $\rho < 1$.
		
		For $1/2$ to be stable we must have $|f'(1/2)| < 1$. By Lemma~\ref{lem:definition_derivative_f} we get:
		\begin{align}
			& |f'(1/2)| < 1 \iff \left|  \frac{2 \lambda \rho + \rho + 1}{(1 + \lambda )(\rho+1)} \right| < 1 \notag \\
			& \iff \frac{2 \lambda  \rho + \rho + 1}{(1 + \lambda )(\rho+1)} < 1 \explain{Since $\lambda ,\rho > 0.$} \\
			& \iff 2 \lambda  \rho + \rho + 1 < \rho + 1 + \lambda  \rho + \lambda \iff \rho < 1\,. \label{eq:derivation_stability_0.5_fixed_point} 
		\end{align}
		The inequalities in \eqref{eq:derivation_stability_0.5_fixed_point} are reversed when $\rho > 1$, which shows  the fixed point $1/2$ is unstable when $\rho > 1$.
	\end{proof}
	
	\section{Conclusion and Future Work}
	We propose a stochastic model and a deterministic approximation to capture the connection between model growth (recruitment) and the network structure. 
	
	There are multiple refinements that would be of further interest. For instance, we can consider a generalization of the dynamic where at step $1$ of the random system in Definition~\ref{def:stochastic} the graph ${G}_{t-1}^+$ is obtained from ${G}_{t-1}$ by refreshing independently the edges of ${G}_{t-1}$ with some probability $q_t \in [0,1]$. This would simulate projects that may run for more than one time step.  
	Furthermore, different time scales for observing the network could be considered.  For example, if $q_t$ is inverse polynomial in the current population size, then the number $m_t$ of arriving nodes could be constant (e.g., $m_t=1$). This would represent the arrival process  as each node joins the organization.
	
	Another generalization is considering multiple communities and allowing death or departure of nodes. Collaborations among larger teams could also be represented using hypergraphs via  a high-dimensional version of the stochastic block model. Collaborations across disciplines (e.g.  where the communities are computer science and physics) could also be captured in this framework.
	
	While the model is stylized, it brings to light that long term interactions across the network must be considered in achieving sustainable diversity in the network. Fig.~\ref{fig:deterministic__random_side_by_side_vanish} shows an example of how even a substantially large minority community can disappear if there is a low fraction of cross-community edges. Thus a recruitment-focused approach to increase diversity in an organization may fail, if sufficient efforts are not taken to integrate minority members into the population.
	
	Finally, while the edge probabilities  $\vec{p}$ might be difficult to change and are determined by individuals, an organization may be able to tweak the values of the weight matrix $\boldsymbol{\zeta}$ to nudge the population in the  direction of achieving parity when it is an important objective.

	\bibliographystyle{alpha}
	\bibliography{literature}

\newcommand{\etalchar}[1]{$^{#1}$}
\begin{thebibliography}{GKM{\etalchar{+}}19}

\bibitem[Abb17]{abbe2017}
Emmanuel Abbe.
\newblock Community detection and stochastic block models: recent developments.
\newblock {\em The Journal of Machine Learning Research}, 18(1):6446--6531,
  2017.

\bibitem[AJC13]{aicher2013adapting}
Christopher Aicher, Abigail~Z Jacobs, and Aaron Clauset.
\newblock Adapting the stochastic block model to edge-weighted networks.
\newblock {\em arXiv preprint arXiv:1305.5782}, 2013.

\bibitem[AJC14]{AJC14}
Christopher Aicher, Abigail~Z. Jacobs, and Aaron Clauset.
\newblock {Learning latent block structure in weighted networks}.
\newblock {\em Journal of Complex Networks}, 3(2):221--248, 06 2014.

\bibitem[AMK{\etalchar{+}}16]{agarwal2016women}
Swati Agarwal, Nitish Mittal, Rohan Katyal, Ashish Sureka, and Denzil Correa.
\newblock Women in computer science research: What is the bibliography data
  telling us?
\newblock {\em Acm Sigcas Computers and Society}, 46(1):7--19, 2016.

\bibitem[Art94]{B94}
W.~Brian Arthur.
\newblock {\em Increasing Returns and Path Dependence in the Economy}.
\newblock University of Michigan Press, 1994.

\bibitem[BA99]{barabasi1999emergence}
Albert-L{\'a}szl{\'o} Barab{\'a}si and R{\'e}ka Albert.
\newblock Emergence of scaling in random networks.
\newblock {\em science}, 286(5439):509--512, 1999.

\bibitem[BGLW18]{bradley2018impact}
Steven~W Bradley, James~R Garven, Wilson~W Law, and James~E West.
\newblock The impact of chief diversity officers on diverse faculty hiring.
\newblock Technical report, National Bureau of Economic Research, 2018.

\bibitem[BIS17]{BIS17}
Eric Balkanski, Nicole Immorlica, and Yaron Singer.
\newblock The importance of communities for learning to influence.
\newblock In {\em NeurIPS}, volume~30, 2017.

\bibitem[BJN{\etalchar{+}}02]{barabasi2002evolution}
Albert-Laszlo Barab{\^a}si, Hawoong Jeong, Zoltan N{\'e}da, Erzsebet Ravasz,
  Andras Schubert, and Tamas Vicsek.
\newblock Evolution of the social network of scientific collaborations.
\newblock {\em Physica A: Statistical mechanics and its applications},
  311(3-4):590--614, 2002.

\bibitem[BMN18]{BMN18}
Simina Br{\^{a}}nzei, Ruta Mehta, and Noam Nisan.
\newblock Universal growth in production economies.
\newblock In {\em NeurIPS}, page 1975, 2018.

\bibitem[Bou18]{bourdieu2018forms}
Pierre Bourdieu.
\newblock The forms of capital.
\newblock In {\em The sociology of economic life}, pages 78--92. Routledge,
  2018.

\bibitem[BS16]{barocas2016big}
Solon Barocas and Andrew~D Selbst.
\newblock Big data's disparate impact.
\newblock {\em California law review}, pages 671--732, 2016.

\bibitem[CAI05]{calvo2005social}
Antoni Calv{\'o}-Armengol and Yannis~M Ioannides.
\newblock Social networks in labor markets.
\newblock Technical report, Department of Economics, Tufts University, 2005.

\bibitem[CAPZ09]{calvo2009peer}
Antoni Calv{\'o}-Armengol, Eleonora Patacchini, and Yves Zenou.
\newblock Peer effects and social networks in education.
\newblock {\em The review of economic studies}, 76(4):1239--1267, 2009.

\bibitem[CCL13]{collevecchio2013preferential}
Andrea Collevecchio, Codina Cotar, and Marco LiCalzi.
\newblock On a preferential attachment and generalized p{\'o}lya’s urn model.
\newblock {\em The Annals of Applied Probability}, 2013.

\bibitem[CPT23]{cullen2023old}
Zoe Cullen and Ricardo Perez-Truglia.
\newblock The old boys’ club: Schmoozing and the gender gap.
\newblock {\em American Economic Review}, 113(7):1703--1740, 2023.

\bibitem[DR01]{DR01}
Pedro Domingos and Matt Richardson.
\newblock Mining the network value of customers.
\newblock In {\em KDD}, pages 57--66. ACM Press, 2001.

\bibitem[DSV12]{dixit2012finite}
Narendra~M Dixit, Piyush Srivastava, and Nisheeth~K Vishnoi.
\newblock A finite population model of molecular evolution: Theory and
  computation.
\newblock {\em Journal of Computational Biology}, 19(10):1176--1202, 2012.

\bibitem[EFN{\etalchar{+}}18]{ensign2018runaway}
Danielle Ensign, Sorelle~A Friedler, Scott Neville, Carlos Scheidegger, and
  Suresh Venkatasubramanian.
\newblock Runaway feedback loops in predictive policing.
\newblock In {\em FAT*}, pages 160--171. PMLR, 2018.

\bibitem[Eig71]{eigen71}
Manfred Eigen.
\newblock Self organization of matter and the evolution of biological
  macromolecules.
\newblock {\em Naturwissenschaften}, 58:1432--1904, 1971.

\bibitem[EP23]{EP23}
F.~Eggenberger and G.~Polya.
\newblock Uber die statistik vorketter vorgange zeit.
\newblock {\em Angew. Math. Mech.}, 3:279--289, 1923.

\bibitem[Ger22]{gertsberg2022unintended}
Marina Gertsberg.
\newblock The unintended consequences of\# metoo-evidence from research
  collaborations.
\newblock {\em Available at SSRN 4105976}, 2022.

\bibitem[GKM{\etalchar{+}}19]{GKMPP19}
Kurtulus Gemici, Elias Koutsoupias, Barnab{\'e} Monnot, Christos~H.
  Papadimitriou, and Georgios Piliouras.
\newblock {Wealth Inequality and the Price of Anarchy}.
\newblock In {\em STACS}, volume 126, pages 31:1--31:16, 2019.

\bibitem[GUSA05]{guimera2005team}
Roger Guimera, Brian Uzzi, Jarrett Spiro, and Luis A~Nunes Amaral.
\newblock Team assembly mechanisms determine collaboration network structure
  and team performance.
\newblock {\em Science}, 308(5722):697--702, 2005.

\bibitem[HLL83]{HLL83}
Paul~W. Holland, Kathryn~Blackmond Laskey, and Samuel Leinhardt.
\newblock Stochastic blockmodels: First steps.
\newblock {\em Social Networks}, 5(2):109--137, 1983.

\bibitem[HLS80]{HLS80}
Bruce~M. Hill, David Lane, and William Sudderth.
\newblock A strong law for some generalized urn processes.
\newblock {\em Ann. Probab.}, 8(2):214--226, 1980.

\bibitem[HS03]{hofbauer2003evolutionary}
Josef Hofbauer and Karl Sigmund.
\newblock Evolutionary game dynamics.
\newblock {\em Bulletin of the American mathematical society}, 40(4):479--519,
  2003.

\bibitem[HZLG08]{huang2008collaboration}
Jian Huang, Ziming Zhuang, Jia Li, and C~Lee Giles.
\newblock Collaboration over time: characterizing and modeling network
  evolution.
\newblock In {\em WSDM}, pages 107--116, 2008.

\bibitem[JFRT16]{BFRT16}
Bo~Jiang, Daniel~R Figueiredo, Bruno Ribeiro, and Don Towsley.
\newblock On the duration and intensity of competitions in nonlinear p{\'o}lya
  urn processes with fitness.
\newblock {\em ACM SIGMETRICS Performance Evaluation Review}, 44(1):299--310,
  2016.

\bibitem[JW03]{jackson2003strategic}
Matthew~O Jackson and Asher Wolinsky.
\newblock A strategic model of social and economic networks.
\newblock In {\em Networks and groups}, pages 23--49. Springer, 2003.

\bibitem[KE12]{KE12}
Jon Kleinberg and David Easley.
\newblock {\em Networks, Crowds, and Markets}.
\newblock Cambridge University Press, 2012.

\bibitem[KK01]{KK01}
K~Khanin and R~Khanin.
\newblock A probabilistic model for the establishment of neuron polarity.
\newblock {\em J Math Biol}, 42:26--40, 2001.

\bibitem[KKT03]{KKT03}
David Kempe, Jon Kleinberg, and \'{E}va Tardos.
\newblock Maximizing the spread of influence through a social network.
\newblock In {\em KDD}, page 137–146, 2003.

\bibitem[KMR16]{kleinberg2016inherent}
Jon Kleinberg, Sendhil Mullainathan, and Manish Raghavan.
\newblock Inherent trade-offs in the fair determination of risk scores.
\newblock In {\em {ITCS}}, 2016.

\bibitem[KRT{\etalchar{+}}23]{kohli2023inclusive}
Sumer Kohli, Neelesh Ramachandran, Ana Tudor, Gloria Tumushabe, Olivia Hsu, and
  Gireeja Ranade.
\newblock Inclusive study group formation at scale.
\newblock In {\em Proceedings of the 54th ACM Technical Symposium on Computer
  Science Education V. 1}, pages 11--17, 2023.

\bibitem[LACL19]{li2019early}
Weihua Li, Tomaso Aste, Fabio Caccioli, and Giacomo Livan.
\newblock Early coauthorship with top scientists predicts success in academic
  careers.
\newblock {\em Nature communications}, 10(1):5170, 2019.

\bibitem[LDR{\etalchar{+}}18]{liu2018delayed}
Lydia~T Liu, Sarah Dean, Esther Rolf, Max Simchowitz, and Moritz Hardt.
\newblock Delayed impact of fair machine learning.
\newblock In {\em International Conference on Machine Learning}, pages
  3150--3158. PMLR, 2018.

\bibitem[LEN18]{ludkin2018dynamic}
Matthew Ludkin, Idris Eckley, and Peter Neal.
\newblock Dynamic stochastic block models: parameter estimation and detection
  of changes in community structure.
\newblock {\em Statistics and Computing}, 28:1201--1213, 2018.

\bibitem[LM88]{lowry1988blot}
Stella Lowry and Gordon Macpherson.
\newblock A blot on the profession.
\newblock {\em British medical journal (Clinical research ed.)}, 296(6623):657,
  1988.

\bibitem[LWH{\etalchar{+}}20]{liu2020disparate}
Lydia~T Liu, Ashia Wilson, Nika Haghtalab, Adam~Tauman Kalai, Christian Borgs,
  and Jennifer Chayes.
\newblock The disparate equilibria of algorithmic decision making when
  individuals invest rationally.
\newblock In {\em Facct}, pages 381--391, 2020.

\bibitem[Mah08]{mahmoud2008polya}
Hosam Mahmoud.
\newblock {\em P{\'o}lya urn models}.
\newblock CRC press, 2008.

\bibitem[mat]{mathematica}
Wolfram mathematica.
\newblock Accessed Dec 2022.

\bibitem[Mor58]{moran1958random}
Patrick Alfred~Pierce Moran.
\newblock Random processes in genetics.
\newblock In {\em Mathematical proceedings of the cambridge philosophical
  society}, volume~54, pages 60--71. Cambridge University Press, 1958.

\bibitem[MPB{\etalchar{+}}21]{mitchell2021algorithmic}
Shira Mitchell, Eric Potash, Solon Barocas, Alexander D'Amour, and Kristian
  Lum.
\newblock Algorithmic fairness: Choices, assumptions, and definitions.
\newblock {\em Annual Review of Statistics and Its Application}, 8:141--163,
  2021.

\bibitem[MSLC01]{mcpherson2001birds}
Miller McPherson, Lynn Smith-Lovin, and James~M Cook.
\newblock Birds of a feather: Homophily in social networks.
\newblock {\em Annual review of sociology}, 27(1):415--444, 2001.

\bibitem[Nat]{UN_inequality}
United Nations.
\newblock Inequality - bridging the divide.
\newblock https://www.un.org/en/un75/inequality-bridging-divide; Accessed July
  2023.

\bibitem[Now06]{nowak2006evolutionary}
Martin~A Nowak.
\newblock {\em Evolutionary dynamics: exploring the equations of life}.
\newblock Harvard university press, 2006.

\bibitem[Pei18]{Peixoto18}
Tiago~P. Peixoto.
\newblock Nonparametric weighted stochastic block models.
\newblock {\em Phys. Rev. E}, 97:012306, Jan 2018.

\bibitem[Pem90]{Pemantle90}
Robin Pemantle.
\newblock A time-dependent version of pólya's urn.
\newblock {\em Journal of Theoretical Probability}, 3, 1990.

\bibitem[Pem07]{Pemantle2007}
Robin Pemantle.
\newblock A survey of random processes with reinforcement.
\newblock {\em Probability Surveys [electronic only]}, 4:1--79, 2007.

\bibitem[PG15]{PG15}
Thomas Piketty and Arthur Goldhammer.
\newblock {\em The Economics of Inequality}.
\newblock Harvard University Press, 2015.

\bibitem[Pik14]{Piketty14}
Thomas Piketty.
\newblock {\em Capital in the Twenty-First Century}.
\newblock Harvard University Press, 2014.

\bibitem[PM95]{petersen1995separate}
Trond Petersen and Laurie~A Morgan.
\newblock Separate and unequal: Occupation-establishment sex segregation and
  the gender wage gap.
\newblock {\em American Journal of Sociology}, 101(2):329--365, 1995.

\bibitem[Pol31]{Polya31}
G.~Polya.
\newblock Sur quelques points de la theorie des probabilites.
\newblock {\em Ann. Inst. H. Poincare}, 1:117--161, 1931.

\bibitem[PSV16]{PSV16}
Ioannis Panageas, Piyush Srivastava, and Nisheeth~K. Vishnoi.
\newblock Evolutionary dynamics in finite populations mix rapidly.
\newblock In {\em SODA}, pages 480--497, 2016.

\bibitem[Put00]{putnam2000bowling}
Robert~D Putnam.
\newblock {\em Bowling alone: The collapse and revival of American community}.
\newblock Simon and Schuster, 2000.

\bibitem[SDG22]{salem2022don}
Jad Salem, Deven Desai, and Swati Gupta.
\newblock {Don’t let Ricci v. DeStefano hold you back: A bias-aware legal
  solution to the hiring paradox}.
\newblock In {\em Facct}, pages 651--666, 2022.

\bibitem[SHC20]{SHC20}
Ana{-}Andreea Stoica, Jessy~Xinyi Han, and Augustin Chaintreau.
\newblock Seeding network influence in biased networks and the benefits of
  diversity.
\newblock In {\em WWW}, pages 2089--2098. {ACM} / {IW3C2}, 2020.

\bibitem[SHL21]{suhr2021does}
Tom S{\"u}hr, Sophie Hilgard, and Himabindu Lakkaraju.
\newblock Does fair ranking improve minority outcomes? understanding the
  interplay of human and algorithmic biases in online hiring.
\newblock In {\em AIES}, pages 989--999, 2021.

\bibitem[SMSL14]{smith2014social}
Jeffrey~A Smith, Miller McPherson, and Lynn Smith-Lovin.
\newblock Social distance in the united states: Sex, race, religion, age, and
  education homophily among confidants, 1985 to 2004.
\newblock {\em American Sociological Review}, 79(3):432--456, 2014.

\bibitem[Tha97]{tharenou1997explanations}
Phyllis Tharenou.
\newblock Explanations of managerial career advancement.
\newblock {\em Australian psychologist}, 32(1):19--28, 1997.

\bibitem[XH14]{xu2014dynamic}
Kevin~S Xu and Alfred~O Hero.
\newblock Dynamic stochastic blockmodels for time-evolving social networks.
\newblock {\em IEEE Journal of Selected Topics in Signal Processing},
  8(4):552--562, 2014.

\end{thebibliography}
	
		\appendix

	\section{Appendix}
	
	In this section we prove Lemma~\ref{lem:expected_ratio_approx_ratio_expectations}. 
	
	\begin{lemma}[Chernoff bound] \label{lem:chernoff}
		Suppose $X_1, \ldots, X_n$ are independent random variables taking values in $\{0, 1\}$. Let $X = \sum_{i=1}^n X_i$ and $\mu = \Ex[X]$.
		Then for  each $\delta  \in (0,1)$,
		\begin{align}
			\Pr(|X - \mu| \geq \delta \mu) \leq 2 e^{-\delta^2 \mu/3}\,.
		\end{align}
	\end{lemma}
	
	\noindent \textbf{Lemma~\ref{lem:expected_ratio_approx_ratio_expectations} (restated).}
	\emph{Let $\epsilon \in (0, 1/2)$. Let  $V = \{1, \ldots, n\}$ be a set of vertices such that   $n_R$ are red and $n_B = n - n_R$  are blue. For a random graph $G$ generated over  $V$ using the stochastic block model with probability matrix $\vec{p}=  \begin{bmatrix}
			{a} & {b} \\
			{b} & {a}
		\end{bmatrix}$ and weight matrix $ 
		\boldsymbol{\zeta}=  \begin{bmatrix}
			\alpha & \beta  \\
			\beta & \alpha 
		\end{bmatrix}$, let $R$  and $B$ be the weight of red and blue, respectively.}
	
	\emph{Suppose $n  \geq C_0 = C(\epsilon, a, \alpha, b, \beta)$ 
		and  $\min\{n_R, n_B\} \geq n^{\frac{1}{2}+\epsilon}$.}
	\emph{Then 
		\begin{align}
			\left(1 - \frac{1}{n^{\epsilon/4}}\right) \frac{\Ex[{R} ]}{\Ex[{R} + {B} ]}   < \frac{{R}}{{R} + {B}} <  \left(1 + \frac{1}{n^{\epsilon/4}}\right) \frac{\Ex[{R}]}{\Ex[{R} + {B}]} \,.
		\end{align}
		with  probability at least $1 - \frac{8}{e^{C_1 \cdot n^{\epsilon}}}$, where 
		$C_1= \min\left\{ {a}/{12}, {b}/{6}\right\}\,.$}
	\emph{Moreover, 
		\begin{align}
			\left(1 - \frac{1}{n^{\epsilon/5}}\right) \frac{\Ex[{R}]}{\Ex[{R} + {B}]} & < \Ex\left[\frac{{R}}{{R} + {B}} \right] <  \left(1 + \frac{1}{n^{\epsilon/5}}\right) \frac{\Ex[{R}]}{\Ex[{R} + {B}]} \,.
	\end{align}}
	\begin{proof}
		W.l.o.g.,  red is the minority in $G$, that is, we have $n_{R} \leq n_{B}$. Then the lemma statement requires  $n_{R} \geq n^{\frac{1}{2}+\epsilon}$. 
		
		Let $\vec{w}$ be the edge weights in $G$, where the weight of each edge $(i,j)$ is $w_{i,j} \in \{0, \zeta_{c_i, c_j}\}$.
		
		Define the  random variables:
		\begin{itemize}
			\item  $RR = \sum_{i,j \in [n]: c_i = c_j = R} w_{i,j}$, the  sum of  weights of the edges with both endpoints red in ${G}$.
			\item $RB = \sum_{i \in [n]: c_i  = R} \sum_{j \in [n]: c_j = B} w_{i,j}$, the sum of weights of  edges with one endpoint red and the other blue in ${G}$.
			\item $BB = \sum_{i,j \in [n]: c_i = c_j = B} w_{i,j}$, the sum of weights of  edges with both endpoints blue in ${G}$.
		\end{itemize}
		
		Let $Y_{i,j} = 1$ if edge $(i,j)$ has strictly positive weight  in $G$  and $Y_{i,j} = 0$ otherwise.  Then
		\begin{itemize}
			\item $RR / \alpha  = \sum_{i,j \in [n]:c_i=c_j=R}  Y_{i,j}$ and  $ {\Ex[RR]}/{\alpha} = \left(n_{R}\right)^2 {a}/{n} $.
			\item $RB / \beta = \sum_{i \in [n]:c_i=R} \sum_{j \in [n]:c_j=B} Y_{i,j} $ and  ${\Ex[RB]}/{\beta} = n_{R} \cdot n_{B} \cdot  {b}/{n}$.
			\item $BB / \alpha = \sum_{i, j \in [n]:c_i=c_j=B}  Y_{i,j} $ and  ${\Ex[BB]}/{\alpha} = \left(n_{B}\right)^2 {a}/{n}$.	
		\end{itemize}  
		
		By definition $R = RR+ RB$ and $B = BB + RB$, so 
		\begin{align} 
			\Ex[{R}] &  = \left(n_{R}\right)^2 \cdot {a \alpha}/{n} + n_R \cdot n_B \cdot {b \beta}/{n} \notag \\  
			\Ex[{B}] & =\left( n_B \right)^2  \cdot {a \alpha}/{n} + n_R \cdot n_B \cdot {b \beta}/{n} \,. \label{eq:def_R_B_expected}
		\end{align}
		
		The Chernoff bound (Lemma~\ref{lem:chernoff}) applied to random variables ${RR}/{\alpha}$, ${RB}/{\beta}$, and ${BB}/{\alpha}$ gives  
		\begin{align}
			\begin{cases}
				& \Pr\left( \left|\frac{RR}{\alpha} - \frac{\Ex[RR]}{\alpha} \right| \geq  \frac{\delta\Ex[RR]}{\alpha}\right)   \leq   
				2 \exp\left( -\frac{\delta^2}{3}   \cdot \left(n_{R}\right)^2 \cdot  \frac{a }{n} \right)   \\
				& \\ 
				& \Pr\left( \left|\frac{RB}{\beta}- \frac{\Ex[RB]}{\beta} \right| \geq \frac{\delta {\Ex[RB]}}{\beta} \right)   \leq 
				2 \exp\left( -\frac{\delta^2}{3}   \cdot n_{R} \cdot n_{B} \cdot \frac{b }{n} \right)   \\
				& \\ 
				& \Pr\left( \left| \frac{BB}{\alpha} - \frac{\Ex[BB]}{\alpha} \right| \geq  \frac{\delta \Ex[BB]}{\alpha}\right)    \leq 
				2 \exp\left( -\frac{\delta^2}{3}   \cdot \left(n_{B}\right)^2 \cdot  \frac{a }{n} \right) \label{eq:RRRBBB_t_chernoff}  \,.
			\end{cases} 
		\end{align}
		
		Since $n_{R} \geq n^{\frac{1}{2}+\epsilon}$, $n_{B} \geq \frac{n}{2}$,  $C_1 = \min\{a/12, b/6\}$, and $\epsilon \in (0,1/2)$, letting $\delta = \frac{1}{n^{\epsilon/2}}$ in \eqref{eq:RRRBBB_t_chernoff} yields  
		\begin{align}
			\Pr\Bigl( |RR - \Ex[RR]| \geq \frac{\Ex[RR]}{n^{\epsilon/2}} \Bigr) & \leq 2 \exp\left( \frac{-a \left(n_{R}\right)^2}{3n^{1 + \epsilon}}     \right)     \leq 2 \exp\left(\frac{-a}{3} \cdot n^{\epsilon } \right)    \leq 2 \exp\left(- C_1 \cdot  n^{\epsilon}  \right) \,. \notag  \\
			\Pr\Bigl( |RB - \Ex[RB]| \geq \frac{\Ex[RB]}{n^{\epsilon/2}} \Bigr)   & \leq  2 \exp\left( \frac{-b n_{R} \cdot n_{B} }{3n^{1+\epsilon}}   \right)   \leq 2 \exp\left( \frac{-b}{6} \cdot n^{\frac{1}{2}} \right)                     \leq 2 \exp\left(- C_1 \cdot  n^{\epsilon}  \right) \notag  \\  
			\Pr\Bigl( |BB - \Ex[BB]| \geq \frac{\Ex[BB]}{n^{\epsilon/2}} \Bigr)  & \leq 2 \exp\left( \frac{-a \left(n_{B}\right)^2}{3n^{1 + \epsilon}}  \right)    \leq 2 \exp\left(\frac{-a}{12} \cdot n^{1-\epsilon} \right)  \leq 2 \exp\left(- C_1 \cdot  n^{\epsilon}  \right) \,. \label{eq:three_inequalities_RR_RB_BB}
		\end{align}
		
		Define the events 
		\begin{align}
			G_{RR} & = \left\{  |RR - \Ex[RR]| \geq \frac{\Ex[RR]}{n^{\epsilon/2}} \right\}\,. \notag \\
			G_{RB} & = \left\{ |RB - \Ex[RB]| \geq \frac{\Ex[RB]}{n^{\epsilon/2}}  \right\} \,. \notag \\
			G_{BB} & = \left\{ |BB - \Ex[BB]| \geq \frac{\Ex[BB]}{n^{\epsilon/2}} \right\} \,.
		\end{align}

		By \eqref{eq:three_inequalities_RR_RB_BB} and  the union bound, 
		\begin{align}
			\begin{cases}
				\Pr\Bigl( G_{RR} \; \; \mbox{or} \;  \; G_{RB} \Bigr)   \leq 4 \exp\left( - C_1 \cdot {n^{\epsilon}}\right)\,.  \\
				\Pr\Bigl( G_{RB} \; \; \mbox{or} \;  \; G_{BB} \Bigr)   \leq 4 \exp\left( - C_1 \cdot {n^{\epsilon}}\right) \,. \label{eq:bounds_combined_RR_RB_and_BB_RB}
			\end{cases}
		\end{align}
		Inequalities \eqref{eq:bounds_combined_RR_RB_and_BB_RB} and \eqref{eq:def_R_B_expected}  imply  
		\begin{align}
			& \Pr\left( |R - \Ex[R]| \geq   \frac{\Ex[R]}{n^{\epsilon/2}} \right)   \leq  4 \exp\left( - C_1 \cdot {n^{\epsilon}}\right) \,. \notag \\
			&   \Pr\left( |B - \Ex[B]| \geq   \frac{\Ex[B]}{n^{\epsilon/2}} \right)   \leq  4 \exp\left( - C_1 \cdot {n_{t}^{\epsilon}}\right)   \,. \label{eq:inequalities_to_bound_A}
		\end{align}
		Then on the event 
		\begin{align}
			\mathcal{A} = \left\{ |R - \Ex[R]| <   \frac{\Ex[R]}{n^{\epsilon/2}}, |B - \Ex[B]| <   \frac{\Ex[B]}{n^{\epsilon/2}} \right\},
		\end{align}
		we have
		\begin{align}
			& 	\left(1 - \frac{1}{n^{\epsilon/2}}\right) \Ex[R] < R < \left(1 + \frac{1}{n^{\epsilon/2}}\right) \Ex[R] \,.  \notag \\
			& \left(1 - \frac{1}{n^{\epsilon/2}}\right) \Ex[R+B] < R  + B   < \left(1 + \frac{1}{n_t^{\epsilon/2}}\right) \Ex[R + B] \,.
		\end{align}
		Suppose $n \geq C_2' = 81^{\frac{1}{\epsilon}}$. 
		Since $\Ex[R + B] > 0$, on event  $\mathcal{A}$ the next inequalities hold:
		\begin{align}
			\frac{\left(1 - \frac{1}{n^{\epsilon/4}}\right)   \Ex[R]}{\Ex[R + B]} & \leq 
			\left( \frac{ 1 - \frac{1}{n^{\epsilon/2}}}{1 + \frac{1}{n^{\epsilon/2}}} \right) \frac{\Ex[R]}{\Ex[R + B]}  \notag \\
			& < \frac{R}{R + B} \notag \\
			&  <
			\left( \frac{ 1 + \frac{1}{n^{\epsilon/2}}}{1 - \frac{1}{n^{\epsilon/2}}} \right) 
			\frac{\Ex[R]}{\Ex[R + B]} 
			\notag \\
			& 			\leq \frac{\left(1 
				+ \frac{1}{n^{\epsilon/4}}\right)  \Ex[R]}{\Ex[R + B]} \,.
		\end{align}
		Thus on event $\mathcal{A}$, the following event occurs:
		\begin{align}
			\mathcal{B} & = \Bigl\{ \Bigr. \left(1 - \frac{1}{n^{\epsilon/4}}\right) \frac{\Ex[R]}{\Ex[R + B]}     < \frac{R}{R + B}   <  \left(1 + \frac{1}{n^{\epsilon/4}}\right) \frac{\Ex[R]}{\Ex[R + B]}  \Bigr\} \,.
		\end{align}
		We have $\mathcal{A} \subseteq \mathcal{B}$, so $1 - \Pr(\mathcal{B}) \leq 1 - \Pr(\mathcal{A})$.  By the union bound, we have 
		\begin{align}
			\Pr(\mathcal{B}^c) = 1 - \Pr(\mathcal{B}) \leq 1 - \Pr(\mathcal{A}) \leq 8 \exp\left(-C_1 \cdot n^{\epsilon}\right), \label{eq:ub_B_complement}
		\end{align}
		where we used \eqref{eq:inequalities_to_bound_A} to upper bound $1-\Pr(\mathcal{A})$. This proves the first part of the lemma. 
		
		Since $\frac{R}{R+B} < 1$, we can upper bound $\Ex\left[\frac{R}{R+B}\right]$ by
		\begin{align}
			\Ex\left[ \frac{R}{R + B}\right] & \leq \Pr(\mathcal{B}) \cdot \left( \left(1 + \frac{1}{n^{\epsilon/4}}\right) \frac{\Ex[R]}{\Ex[R + B]}   \right) + \Pr(\mathcal{B}^c) \cdot 1  \notag \\
			& \leq 1 \cdot  \left( \left(1 + \frac{1}{n^{\epsilon/4}}\right) \frac{\Ex[R]}{\Ex[R + B]}   \right) + 8 \exp\left(-C_1 \cdot n^{\epsilon}\right)\,.  \explain{Using \eqref{eq:ub_B_complement} and $\Pr(\mathcal{B}) \leq 1$}  
		\end{align} 
		We have   
		\begin{align}
			\left(1 + \frac{1}{n^{\epsilon/4}}\right) \frac{\Ex[R]}{\Ex[R + B]} + 8 \exp\left(-C_1 \cdot n^{\epsilon}\right)  \leq \left(1 + \frac{1}{n^{\epsilon/5}}\right) \frac{\Ex[R]}{\Ex[R + B]} \,. \label{eq:required_ub_expected_ratio_red_intermediate_1}  \end{align}
		Inequality \eqref{eq:required_ub_expected_ratio_red_intermediate_1} holds if and only if 
		\begin{align} 
			\left( \frac{1}{n^{\epsilon/5}} - \frac{1}{n^{\epsilon/4}} \right) \frac{\Ex[R]}{\Ex[R+B]} \geq 8 \exp\left(-C_1 \cdot n^{\epsilon}\right)\,. \label{eq:expected_ratio_red_intermediate_ub1}
		\end{align}
		For $n \geq C_2'' = \exp\left(\frac{20 \ln{2}}{\epsilon}\right)$, we have 
		\[ 
		{1}/{n^{\epsilon/5}} - {1}/{n^{\epsilon/4}} \geq {1}/{n^{\epsilon/4}},
		\] so the inequality 
		\begin{align}
			\left( \frac{1}{n^{\epsilon/5}} - \frac{1}{n^{\epsilon/4}} \right) \frac{\Ex[R]}{\Ex[R+B]} \geq  \frac{1}{n^{\epsilon/4}} \cdot \frac{\Ex[R]}{\Ex[R + B]}  \geq 8 \exp\left(-C_1 \cdot n^{\epsilon}\right) 
		\end{align} 
		holds when 
		\begin{align}
			\frac{\Ex[R]}{\Ex[R + B]} & \geq \frac{8 n^{\epsilon/4}}{e^{C_1 \cdot n^{\epsilon}}}\,. \label{eq:expected_ratio_red_intermediate_ub2}
		\end{align}
		By Lemma~\ref{lem:ub_lb:expected_ratio_approx_ratio_expectations}, we have 
		\[ \frac{\Ex[R]}{\Ex[R + B]} \geq \frac{2 b \beta}{\left(5a \alpha  + 4 b \beta\right)\cdot n^{\frac{1}{2}-\epsilon}}\,.
		\] 
		Thus \eqref{eq:expected_ratio_red_intermediate_ub2} would hold if 
		\begin{align}
			& \frac{2 b \beta}{\left(5a \alpha  + 4 b \beta\right)\cdot n^{\frac{1}{2}-\epsilon}} \geq \frac{8 n^{\epsilon/4}}{e^{C_1 \cdot n^{\epsilon}}} \iff \notag  \\
			& \left(\frac{b \beta}{5a\alpha + 4 b \beta} \right) e^{C_1 \cdot n^{\epsilon}} \geq 4 n^{\frac{1}{2} - \frac{3\epsilon}{4}}  \iff \notag \\
			&  n^{\epsilon} \geq \ln(n)  \left( \frac{2 - 3 \epsilon}{4 C_1}\right) + \frac{\ln(2)}{4 C_1}  - \frac{1}{C_1} \ln\left(\frac{b \beta}{5a\alpha + 4 b \beta} \right)  \,. \label{eq:exponential_more_poly_ub3} 
		\end{align} 
		
		There exists a constant $C_2''' = C(\epsilon,a,\alpha,b,\beta)$ so that  \eqref{eq:exponential_more_poly_ub3} holds for all $n \geq C_2'''$. 
		
		Thus  \eqref{eq:required_ub_expected_ratio_red_intermediate_1} holds for $n \geq \max\{C_2', C_2'', C_2'''\}$,  which completes the  upper bound on $\Ex\left[\frac{R}{R+B} \right]$.
		
		To lower bound $\Ex\left[\frac{R}{R+B} \right]$, we write:
		\begin{align}
			\Ex\left[\frac{R}{R+B} \right] & \geq \Pr(\mathcal{B}) \cdot \left( 1 - \frac{1}{n^{\epsilon/4}} \right) \frac{\Ex[R]}{\Ex[R+B]}   \notag \\
			& \geq \left(1 -  \frac{8}{e^{C_1 \cdot n^{\epsilon}}} \right) \left( 1 - \frac{1}{n^{\epsilon/4}} \right) \frac{\Ex[R]}{\Ex[R+B]} \,. \label{eq:lb_expected_ratio_red_intermediate_1} 
		\end{align}
		There exists $C_2'''' = C(\epsilon, a, \alpha, b, \beta)$ such that for $n \geq C_2''''$ we have \[
		\left(1 -  \frac{8}{e^{C_1 \cdot n^{\epsilon}}} \right) \left( 1 - \frac{1}{n^{\epsilon/4}} \right) \geq 1 - 1/n^{\epsilon/5},
		\] 
		thus  
		\[ \Ex\left[\frac{R}{R+B} \right]  \geq \left( 1 -  \frac{1}{n^{\epsilon/5}} \right)  \frac{\Ex[R]}{\Ex[R+B]} \,.
		\]
		
		Take $C_0 =  \max\left(81^{\frac{1}{\epsilon}}, C_2', C_2'', C_2''', C_2''''\right)$. Then for $n \geq C_0$, we have 
		\begin{align} 
			& \left(1 - \frac{1}{n^{\epsilon/5}}\right) \frac{\Ex[{R}]}{\Ex[{R} + {B}]} < \Ex\left[\frac{{R}}{{R} + {B}} \right]  <  \left(1 + \frac{1}{n^{\epsilon/5}}\right) \frac{\Ex[{R}]}{\Ex[{R} + {B}]},
		\end{align}
		as required by the lemma statement.
	\end{proof}
	
	\begin{lemma} \label{lem:ub_lb:expected_ratio_approx_ratio_expectations}
		In the setting of Lemma~\ref{lem:expected_ratio_approx_ratio_expectations}, 
		\begin{align}
			\frac{2 b \beta}{\left(5a \alpha  + 4 b \beta\right)\cdot n^{\frac{1}{2}-\epsilon}}  \leq 	\frac{\Ex[R]}{\Ex[R+B]}  \notag \,.
		\end{align}
	\end{lemma}
	\begin{proof}
		By \eqref{eq:def_R_B_expected}, 
		\begin{align}
			\frac{\Ex[R]}{\Ex[R+B]} & = \frac{{a \alpha}   \left(n_{R}\right)^2   + {b \beta}  \left(  n_{R} \cdot n_{B} \right)}{ {a \alpha}   \left(n_{R} \right)^2   + 2 \cdot  {b \beta}  \left(  n_{R} \cdot n_{B} \right) +  {a \alpha}  \left( n_{B} \right)^2 } \notag \\ 
			& \geq \frac{a \alpha \cdot n^{1 + 2 \epsilon} + b \beta \cdot n^{\frac{1}{2} + \epsilon} \cdot \frac{n}{2}}{a \alpha \cdot  \left(\frac{n}{2}\right)^2 + 2 \cdot b \beta \cdot \frac{n}{2} \cdot n + a \alpha \cdot n^2} \notag \\
			& = \frac{a \alpha \cdot {n^{\epsilon - \frac{1}{2}}} + b \beta/2}{\left(5a \alpha/4  + b \beta\right)\cdot n^{\frac{1}{2}-\epsilon}} \notag \\
			& \geq \frac{2 b \beta}{\left(5a \alpha  + 4 b \beta\right)\cdot n^{\frac{1}{2}-\epsilon}} \,.
		\end{align}
		This completes the proof of the lemma.
	\end{proof}

\end{document}